\documentclass[10pt, conference,letterpaper]{IEEEtran}

\IEEEoverridecommandlockouts
\usepackage{authblk}
\usepackage{cite}
\usepackage{amssymb,amsfonts,dsfont}
\usepackage{algorithm}
\usepackage[noend]{algorithmic}
\usepackage{graphicx}
\usepackage{textcomp}
\usepackage{xcolor}
\usepackage{bm}
\usepackage{booktabs}
\usepackage{multirow}
\usepackage{lipsum}
\usepackage{url}
\usepackage{amsthm}
\usepackage{threeparttable}
\usepackage{makecell}

\usepackage[utf8]{inputenc} 
\usepackage[T1]{fontenc}
\usepackage{ifthen}
\usepackage[cmex10]{amsmath} 

\newtheorem{theorem}{Theorem}

\newtheorem{remark}{Remark}

\renewcommand{\algorithmicrequire}{\textbf{Input:}} 
\renewcommand{\algorithmicensure}{\textbf{Output:}}
 
\DeclareMathOperator*{\argmin}{\arg\min}

\newcommand{\bdphi}{\boldsymbol{\phi}} 
\newcommand{\bdpsi}{\boldsymbol{\psi}}


\begin{document}
\title{Alternating Maximization Algorithm for Mismatch Capacity with Oblivious Relaying
\thanks{The first three authors contributed equally to this work and $\dag$ marked the corresponding author. This work was partially supported by National Key Research and Development Program of China (2018YFA0701603) and National Natural Science Foundation of China (12271289 and 62231022).}
} 

\author[1]{Xinwei Li}
\author[1]{Lingyi Chen}
\author[1]{Shitong Wu}

\author[2$\dag$]{Huihui Wu}
\author[1]{Hao Wu}
\author[3]{Wenyi Zhang}

\affil[1]{Department of Mathematical Sciences, Tsinghua University, Beijing 100084, China}
\affil[3]{Department of Electronic Engineering and Information Science, 
\authorcr University of Science and Technology of China, Hefei, Anhui 230027, China 
}
\affil[2]{Yangtze Delta Region Institute (Huzhou), 
\authorcr University of Electronic Science and Technology of China, 
Huzhou, Zhejiang, 313000, China. 
\authorcr Email: huihui.wu@ieee.org
}

\maketitle


\begin{abstract}
Reliable communication over a discrete memoryless channel with the help of a relay has aroused interest due to its widespread applications in practical scenarios. 
By considering the system with a mismatched decoder, previous works have provided optimization models to evaluate the mismatch capacity in these scenarios. 
The proposed models, however, are difficult due to the complicated structure of the mismatched decoding problem with the information flows in hops given by the relay. 
Existing methods, such as the grid search, become impractical as they involve finding all roots of a nonlinear system, with the growing size of the alphabet. 
To address this problem, we reformulate the max-min optimization model as a consistent maximization form, by considering the dual form of the inner minimization problem and the Lagrangian with a fixed multiplier. 
Based on the proposed formulation, an alternating maximization framework is designed, which provides the closed-form solution with simple iterations in each step by introducing a suitable variable transformation. 
The effectiveness of the proposed approach is demonstrated by the simulations over practical scenarios, including Quaternary and Gaussian channels. 
Moreover, the simulation results of the transitional probability also shed light on the promising application attribute to the quantizer design in the relay node.
\end{abstract}

\section{Introduction}
%


%
Reliable communication over a discrete memoryless channel (DMC) with a relay has aroused considerable interest \cite{Meulen1971THREETERMINALCC}, due to its widespread applications in practical scenarios \cite{Quek_Peng_Simeone_Yu_2017, Wang2020AnOO, Guido2011Quantize}. 
In communications with relaying, a transmitter sends a message to one or multiple oblivious agencies, and the message is then conveyed to the receiver over links with restricted capacities \cite{Dikshtein2023OnMO, Hucher2010MismatchedDF}. 
The uplink Cloud Radio Access Network (C-RAN) is a typical application of this model and extensive investigations on the capacity of this scenario have been concerning the multi-user communications with base stations \cite{EstellaAguerri2019OnTC, Chen2016AlternatingIB, Xu2021InformationBF}. 
The capacity of the channels with oblivious relaying are commonly modeled by the Information Bottleneck (IB) framework \cite{Tishby2000TheIB}, due to the similarity between the constraints of two optimization problems \cite{Dikshtein2023OnMO}. 
Specifically, the compression quality under the IB approach is evaluated by the mutual information between the compressed representation and the bottleneck variable, which inherently matches 
the channel capacity with oblivious relaying. 
In this work, we focus on channels with oblivious relaying, in the case of mismatched decoding \cite{Hucher2010MismatchedDF, Dikshtein2023OnMO}. 
Recent work \cite{Dikshtein2023OnMO} derives a random coding capacity of the discrete memoryless information bottleneck channel (DM-IBC) with a mismatched decoder and the corresponding achievable bound of random codes is also analyzed. 
Specifically, the theoretical capacity of such a channel is proved to be objected by the LM (Lower [bound on the] Mismatch [capacity]) rate \cite{csiszar1981graph,Scarlett2020Information}, which is the tightest achievable rate via a constant-composition codebook ensemble \cite{Merhav1994Mismatched,Csiszar1995Channel,csiszar2011information}, and constrained by the mutual information with upper bound similar to the IB problem. 
The optimization models of the mismatch capacity with oblivious relaying are complicated due to the utilization of the LM rate as the objective function in the IB problem, instead of the mutual information. 
This implies that algorithms for the IB problem \cite{Tishby2000TheIB, blahut_1972_computation, Chen23IBOT, chen2023srib}, cannot be commonly employed.
Moreover, the numerical methods for the LM rate \cite{ye2022optimal,Scarlett2020Information} only provide a minimization solution of the inner problem for this optimization models, which will lead to inefficiency when directly applied. 
To solve this problem, the authors of \cite{Dikshtein2023OnMO} proposed an algorithm named MMIB, resembling an adaptive grid search, which requires finding all roots of a nonlinear system constrained by the mutual information. 
As the size of the input alphabet increases, this method will become intractable.  
To address this problem, this paper develops an alternating maximization algorithm to efficiently compute the mismatch capacity with oblivious relaying. 
Specifically, to handle the optimization constraint, an IB-Lagrangian model with a fixed multiplier is proposed, by considering the Lagrange relaxation of the original problem. 
Additionally, we explore a dual formulation of the LM rate \cite{ye2022optimal} from the optimal transport perspective, which transforms the optimization into a consistent maximization problem. 
To ensure an analytic form in each step of the algorithm, a variable transformation is further introduced. 
It should be noted that the fixed multiplier plays a role in the convergence of the algorithm, as will be discussed subsequently.
Finally, the numerical experiments validate the efficiency and accuracy of the proposed algorithm. 
The results of the optimized probability distributions shed light on the promising application to the quantizer in practical scenarios.

\section{Problem Formulation} \label{sec_problem}
We consider a 3-node point-to-point communication system with a relay depicted in Fig. \ref{fig:relay}, referred to as the discrete memoryless information bottleneck channel (DM-IBC) \cite{Dikshtein2023OnMO}.  
\begin{figure}[H]
	\centering
		\includegraphics[width=0.95\linewidth]{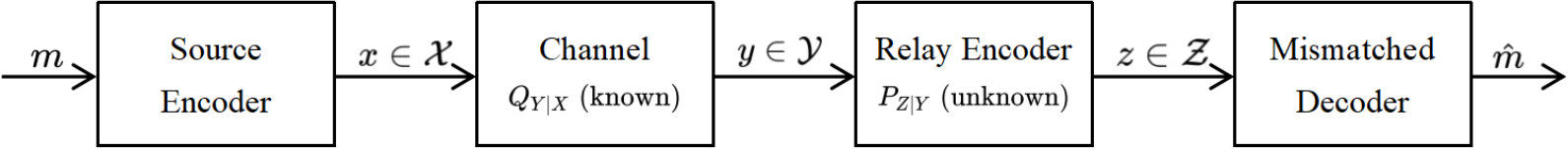}
        \caption{A 3-node channel with oblivious relaying and a mismatched decoder}
	\label{fig:relay}
\end{figure}
Under its model presuppositions, we denote the input alphabet $\mathcal{X} = \{x_i\}_{i=1}^M$, the received alphabet $\mathcal{Y} = \{y_k\}_{k=1}^K$, and the reproduced alphabet $\mathcal{Z} = \{z_j\}_{j=1}^N$. 
The conditional probability distribution between $\mathcal{X}$ and $\mathcal{Y}$ is denoted by $Q_{Y|X}(y|x)$, which is only known by the encoder. 
With an input distribution $P_X \in \mathcal{P(X)}$, the output distribution at the relay node $P_Y \in \mathcal{P(Y)}$ and the joint distribution $P_{XY} \in \mathcal{P(X \times Y)}$ are derived from the transition probability $Q_{Y|X}$. 
After compression and transmission though an unknown relay encoder ${P}_{Z|Y}(z|y)$, the joint distribution ${P}_{XZ} \in \mathcal{P}(\mathcal{X} \times \mathcal{Z})$ is given by 
\vspace{.01in}
\begin{equation} \label{Pxz}
    P_{XZ}(x,z) = \sum_{y\in\mathcal{Y}} P_{Z|Y}(z|y)Q_{Y|X}(y|x)P_X(x).
\end{equation}
%

Given a transmission rate $R$ and a compression rate $B$,  the encoding scheme at the sender is described by a codebook of $2^{nR}$ $n$-dimensional sequence $x^{n}(m)\in\mathcal{X}^n$, 
while the encoding scheme at the relay source encoder is described by a codebook of $2^{nB}$ $n$-dimensional sequence $z^{n}(w)\in\mathcal{Z}^n$, 
When a message $m$ is chosen uniformly from the message set $\mathcal{M}=\{1, 2, \cdots, 2^{nR}\}$, the sender transmits the corresponding codeword $x^{n}(m)$ to the relay source encoder.
The source encoder at the relay receives $y^n\in\mathcal{Y}^n$, and then compresses it into $z^{n}(w)$, where $w \in \mathcal{W}=\{1, 2, \cdots, 2^{nB}\}$ 
is a representation set.
The mismatched decoder then provides an estimate $\hat{m}\in\mathcal{M}$ of the intended transmission message based on the reproduced sequence $z^{n}(w)$
and the following prescribed decoding rule
\begin{equation*}
     \hat{m} = \argmin_{j\in\mathcal{M}}\sum_{i=1}^{n}d(x_{i}(j),z_{i}),
\end{equation*}
where $d: \mathcal{X} \times \mathcal{Z}\to\mathbb{R}$ is called the mismatched metric \cite{Merhav1994Mismatched}. 
With these notations, the random coding capacity of the DM-IBC with a mismatched decoder is defined as \cite{Dikshtein2023OnMO}, 
\vspace{.05in}
\begin{equation} \label{capacity_def}
    \begin{aligned}
        C_d(B) \triangleq \max_{P_{X}}
        \max_{{P}_{Z|Y}}\quad& I_{\mathrm{LM}}(X;Z) \\
        \mathrm{s.t. }\quad & I(Y; Z) \leq B.
    \end{aligned}
\end{equation}
%
Here $I_{\mathrm{LM}}(X;Z)$ is the LM rate defined via $P_{XZ}$ \eqref{Pxz}, and is 
\newpage
\noindent
actually an optimization problem
\addtolength{\topmargin}{.02in}
\begin{equation} \label{rate_def}
I_{\mathrm{LM}}(P_X, P_{Z|X}) =\min_{\substack{\widetilde{P}_{X Z} \in \mathcal{P}(\mathcal{X} \times \mathcal{Z}): \widetilde{P}_X=P_X, \widetilde{P}_Z=P_Z \\ \mathbb{E}_{\widetilde{P}_{XZ}}[d(X, Z)] \leq \mathbb{E}_{P_{XZ}}[d(X, Z)]}} I_{\widetilde{P}}(X ; Z).
\end{equation}
Here $I_{LM}(P_X, P_{Z|X})$ refers to the LM rate $I_{\mathrm{LM}}(X;Z)$ since its value is determined by distributions $P_X$ and $P_{Z|X}$.
The notation $\widetilde{P}$ is the auxiliary probability distribution corresponding to the variables to be optimized in \eqref{rate_def} and the index $d$ in $C_d(B)$ corresponds to the decoding metric $d(x,z)$. 
It is noted that the optimization \eqref{capacity_def} presents the similar structure as the IB problem, with substituting the objective function with the LM rate, 
which imposes a considerable challenge on solving this optimization problem. 
The LM rate \eqref{rate_def} here is different from that in the usual sense, primarily because the distribution $P_{XZ}$ has to be treated as a variable and is also linked via $P_X$ and ${P}_{Z|Y}$ using \eqref{Pxz}. 
Moreover, the problem \eqref{capacity_def} possesses a max-min structure and it is non-convex.
Thus, the solution to problem \eqref{capacity_def} is intricate. 
The MMIB algorithm was proposed in \cite{Dikshtein2023OnMO} for solving the problem \eqref{capacity_def}, including a step which entailed finding all roots of a system of nonlinear equations.
Finding all roots in this system is almost impossible in large-scale cases, and even solving for part of the roots selectively still results in large computational costs. 
It is therefore apparent that a more efficient algorithm for this problem is required.

\section{Alternating Maximization Algorithm} 
\label{sec_main}
In this section, an efficient algorithm will be proposed for solving \eqref{capacity_def}, with a prior reformulation processing of the model.
We simplify the original formulation \eqref{capacity_def} via the IB-Lagrangian type relaxation and the dual form of the LM rate from the optimal transport perspective, which enables us to solve it in an alternating ascend manner. 

\vspace{-.04in}
\subsection{Relaxation and Dual Form}
The compression rate constraint in equation \eqref{capacity_def} involves an unknown variable, $P_{Z|Y}$, which presents a challenge in terms of the direct computation. 
To address this issue, we utilize a technique similar to the IB-Lagrangian in the BA algorithm \cite{Tishby2000TheIB}, by introducing a fixed Lagrange multiplier $\lambda$ corresponding to slope of the curve to evaluate the trade-off between the LM rate and the compression rate $B$, which yields a relaxation of the original problem 
\begin{equation}    \label{capacity_relax}
\max_{P_{X}}\max_{{P}_{Z|Y}} 
\left( I_{\mathrm{LM}}(P_X, P_{Z|X}) - \lambda  I(Y; Z) \right).
\end{equation}
\vspace{-.11in}

Note that the solution of problem \eqref{capacity_relax} is also not straightforward, since it includes the term $I_{\mathrm{LM}}(P_X, P_{Z|X})$ in \eqref{rate_def}. 
In the general case, the objective function of an IB problem is the mutual information with an explicit expression.
However, the LM rate here is an optimization, without an analytic form.

Consequently, we introduce an optimal transport based dual form of the LM rate for fixed $P_{XZ}$ \cite{chen2024double}, i.e.,
\begin{equation*}
    I_{\mathrm{LM}}(P_X, P_{Z|X}) = \max_{\substack{\bm\phi, \bm\psi, \\ \zeta \geq 0}} ~ g_{\mathrm{LM}}(\bm\phi,\bm\psi,\zeta).
\end{equation*}
And the above optimization exhibits a consistent maximization 
\newpage
\noindent
formulation with the objective function 
\begin{equation*}\label{lm_dual}
\begin{aligned}
    g_{\mathrm{LM}} (\bm\phi,\bm\psi,\zeta)
    =
    & -\bm\phi^{T} \bm\Lambda \bm\psi + H(X) + H(Z) + \mathbb{E}_{P_{XZ}}(\log \bm\Lambda) \\
    &  
    + \mathbb{E}_{P_X}(\log \bm\phi) +  \mathbb{E}_{{P}_{Z}}(\log \bm\psi) + 1,
\end{aligned}
\end{equation*}
where $\bm\Lambda = e^{-\zeta d(X, Z)}$, and $H(\cdot)$ denotes the information entropy.
Since $g_{\mathrm{LM}} (\bm\phi,\bm\psi,\zeta)$ is affected by the unknown distribution $P_X$ and $P_{Z|Y}$, we denote it as $g_{\mathrm{LM}} (\bm\phi,\bm\psi,\zeta; P_X, P_{Z|Y})$ in the following.
Accordingly, \eqref{capacity_relax} is transformed to
\begin{equation} \label{capacity_use}
    \max_{\substack{P_X, \\ P_{Z|Y}}}~\max_{\substack{\bm\phi, \bm\psi, \\ \zeta \geq 0}}~g_{\mathrm{LM}}(\bm\phi,\bm\psi,\zeta; P_X, P_{Z|Y}) - \lambda~ \mathbb{E}_{P_{YZ}}\left(\log \frac{P_{Z|Y}}{P_Z}\right),
\end{equation}
where $P_{YZ}(y, z)=P_{Z|Y}(z|y)P_{Y}(y)$ 
is also unknown. 

\subsection{Alternating Maximization Algorithm}
\label{sec:alg}
Prior to this, we introduce some annotations,
\vspace{-.05in}
\begin{alignat*}{3}
    p_i &= {P}_X(x_i), ~  q_{k}\! &={P}_Y(y_k),\,\,\, ~ \Theta_{ik} & = Q_{Y|X}(y_k | x_i), \\
    r_j &= {P}_{Z}(z_j), ~ D_{ij} &= d(x_i, z_j), ~ \Omega_{kj} &= {P}_{Z|Y}(z_j|y_k).
\end{alignat*}
The bold lowercase letter, e.g., $\bm r$, denotes the corresponding vector of the probability distribution, while the bold uppercase letter, e.g., $\bm \Omega$, denotes the corresponding metric. 
The notion $\odot$ represents the pointwise multiplication. 
The notion $\bm q \,\, ./ \,\, \bm p$ represents the division of each element of $\bm q$ by the corresponding element of $\bm p$. 
Following the reformulation of the consistent maximization problem \eqref{capacity_use}, an algorithm for solving it by an alternating ascending approach can be given.
Furthermore, it is highlighted that a closed-form solution can still be obtained at each iteration step, despite the highly complex nature of the problem, which will be addressed in the subsequent discussion.
The objective function in \eqref{capacity_use} contains the entropy term $H(Z) = - (\bm{\Omega}^T \bm{\Theta}^T \bm p)^T \log (\bm{\Omega}^T \bm{\Theta}^T \bm p)$, making it difficult to obtain a closed-form solution for updating $\bm p$.
Thus, we utilize a variable transformation $\bm{\widetilde{\psi}} = \bm{\psi} ./ (\bm{\Omega}^T \bm{\Theta}^T \bm p)$.
Then the relaxation model \eqref{capacity_use} is equivalent to the maximization problem
\begin{equation}\label{capacity_dis}
\begin{aligned}
    \max_{\bm p, \bm t, \bm r}\,\,\max_{\substack{ \bdphi, \widetilde{\bdpsi}, \\ \zeta \geq 0}}& \,\,  - \bm p^T \log \bm p + \bm p^T \log \bm J(\bdphi, \widetilde{\bdpsi}, \zeta; \bm \Omega, \bm r) + 1 \\
    \mathrm{s.t.} &\,\,\,  \bm \Omega \mathbf{1}_{N} = \mathbf{1}_{K},\,\,
    \|\bm r\|_1 = 1,\,  \|\bm p\|_1 = 1.
\end{aligned}
\end{equation}
The function $\bm J(\bdphi, \widetilde{\bdpsi}, \zeta; \bm \Omega, \bm r)$ is defined as
\begin{equation*}
\begin{aligned}
    & \bm J =  \bm \phi \odot \exp \left\{
    (\bm \Theta \bm \Omega)\left[ -\left(\bm \Lambda ^T \bm \phi\right)\odot\widetilde{\bm \psi} 
    +\log \widetilde{\bm\psi}\right]
    \right. \\
    &\left.
    - \zeta \left[(\bm \Theta \bm \Omega)\odot \bm D\right]\mathbf{1}_N
    - \lambda \bm \Theta\left[\bm \Omega \odot \left(\log \bm \Omega - \log \left(\mathbf{1}_K \bm r ^T \right)\right) \right]
    \right\},
\end{aligned}
\end{equation*}
which denotes the sum of coefficients in linear terms related to the variable $\bm p$.


The main derivation of the algorithm is presented as follows.

\subsubsection{Update the input distribution $P_X$}
Consider the optimization \eqref{capacity_dis} with respect to $\bm p$ only,
and construct its Lagrangian with the dual variable $\eta_p\in\mathbb{R}$,
\vspace{-.05in}
\begin{equation*}    \label{lag_p}
    \begin{aligned}
        \mathcal{L}_p  (\bm p; \eta_p) = &
        H(\bm p) + \mathbb{E}_{P_X}\left(\log \bm{J}(\bm\phi, \widetilde{\bm\psi}, \zeta; \bm \Omega, \bm r) \right) + 1 \\
        & - \eta_p\left( \|\bm p\|_1 - 1 \right).
    \end{aligned}
\end{equation*}
By considering the first-order condition and further analyzing 
\newpage
\noindent
its the dual form, we can update variables $p_i$ according to
$p_i = {J_i}/{\sum_{i^{'} = 1}^M J_{i^{'}}}$.
%

\subsubsection{Update the joint distribution between $Y$ and $Z$}
Considering the first-order condition with respect to $\Omega_{kj}$ $ = {P}_{Z|Y}(z_j|y_k)$, we could obtain an optimal solution.
In conjunction with the constraints in \eqref{capacity_dis}, variables $\bm \Omega$ can be updated through the normalized condition corresponding to the conditional probability distribution, that is
$
    \Omega_{kj} = {\Omega_{kj}^*}/{\left(\sum_{j^{'}=1}^N \Omega_{kj^{'}}^*\right)},
$
where 

\vspace{-.1in}
\begin{small}
\begin{equation*}
\begin{aligned}
    \bm \Omega^* = 
    & \left\{
    \mathbf{1}_K \left[\bm r 
    \odot (\widetilde{\bm \psi}^{\frac{1}{\lambda}})
    \odot \exp \left( \frac{1}{\lambda} (- \bm \Lambda ^T \bm \phi)\odot\widetilde{\bm \psi} \right) 
    \right]^T\right\}\\
    & \odot 
    \exp \left[-\frac{\zeta}{\lambda} (\bm \Theta^T \operatorname{Diag}(\bm p)\bm D)./(\bm \Theta^T \bm p \mathbf{1}_N^T)\right].
\end{aligned}
\end{equation*}
\end{small}

In a similar manner, variables $r_j = P_Z(z_j)$ can be updated according to $\bm r = \bm \Omega ^T \bm \Theta ^T \bm p$.

\subsubsection{Update the dual variables of the LM rate}For the fixed input distribution $p_i = P_X(x_i)$ and joint distribution $\Omega_{kj} = \widetilde{P}_{Z|Y}(z_j|y_k)$, this process is equivalent to computing the LM rate between $X$ and $Z$ with $P_{XZ}$ derived according to \eqref{Pxz}.
Similar to the previous work \cite{ye2022optimal}, we update the dual variables alternatively by
\begin{equation*}
    \phi_i  = \dfrac{p_i}
    {\sum_{j=1}^N e^{-\zeta D_{ij}}\widetilde{\psi_j} r_j },\quad 
    \widetilde{\psi_j}  = \dfrac{1}
    {\sum_{i=1}^M \phi_{i} e^{-\zeta D_{ij}}}.
\end{equation*}

Also, $\zeta$ is updated by finding the root of the following monotone function

\vspace{-.1in}
\begin{small}
\begin{equation*}\label{upd_ze}
\begin{aligned}
    G(\zeta) \triangleq 
    \bm \phi^T \left( \bm D \odot \bm \Lambda \right) \left[\widetilde{\bm \psi} \odot \left(\bm \Omega ^T \bm \Theta ^T \bm p \right)\right] 
    -
    \bm p^T \left[  \left(\bm \Theta  \bm \Omega \right)
    \odot\bm D \right]\mathbf{1}_N,
\end{aligned}
\end{equation*}
\end{small}
where $\bm\Lambda = e^{-\zeta\bm D}$ as defined before.
At this point, the primary derivation of the alternating maximization process has been described. 
We summarise the proposed algorithm in the pseudo-code in Algorithm \ref{alg:main}.
Due to space limitations, we put the derivation details of the AM algorithm in the appendix.  

\vspace{-.1 in}
\begin{algorithm}[htbp] 
	\renewcommand{\algorithmicrequire}{\textbf{Input:}}
	\renewcommand\algorithmicensure {\textbf{Output:} }
	\caption{Alternating Maximization (AM) Algorithm}
	\label{alg:main}
	
	\begin{algorithmic}[1]
		\REQUIRE 
		Conditional probability distribution $\bm \Theta$, Decoding metric $\bm D$, Lagrange multiplier $\lambda$, Iteration number $max\_iter$. \\
         
        \STATE \textbf{Initialize:} $\bm\phi^{(0)} = \bm 1_M$, $\widetilde{\bm\psi}^{(0)} = \bm 1_N$, $\zeta = 1$, $\mu = 1$; \textbf{Randomly Initialize:}  $\bm \Omega^{(0)}$, $\bm r^{(0)}$.\\

        \FOR{$l$ = 1 : $max\_iter$}
        
        \STATE Update $p_i$ $\leftarrow$ ${J_i }/{\left(\sum_{i^{'} = 1}^M J_{i^{'}} \right)}$
        
        \STATE Update $\Omega_{kj}$ $\leftarrow$ ${\Omega_{kj}^*} / {\left(\sum_{j^{'}=1}^N \Omega_{kj^{'}}^*\right)}$

        \STATE Update $r_j$ $\leftarrow$ $\sum_{i=1}^M \sum_{k = 1}^K  \Omega_{kj} \Theta_{ik} p_i$

        \STATE Update $\phi_i$ $\leftarrow$ 
        ${p_i} / 
        \left({\sum_{j=1}^N e^{-\zeta D_{ij}}\widetilde{\psi_j} r_j}\right)$
      
        \STATE Update $\widetilde{\psi_j}$ $\leftarrow$ 
        ${1}/
        \left({\sum_{i=1}^M \phi_i e^{-\zeta D_{ij}}}\right)$
      
        \STATE  Solve $G(\zeta) = 0$ for $\zeta\in\mathbb{R}^+$ with Newton's method
        
        \ENDFOR
        \RETURN $C_d$
	\end{algorithmic}
\end{algorithm}
\begin{theorem}
    The proposed AM algorithm ensures the convergence of all iteration variables.
\end{theorem}
\begin{proof}
%
In each iteration, the variables are updated individually
\newpage
\noindent
based on their respective first-order conditions.
It is noteworthy that the objective function in \eqref{capacity_dis} is concave  with respect to the variables $\bm{p}$, $\bm{\Omega}$, $\bm{\phi}$, $\bm{\psi}$, and $\zeta$ separately. 
Thus, the first-order condition is equivalent to identifying the maximum value of the objective function \eqref{capacity_use} at each iteration.

The value of the objective function demonstrates a sustained increase through the alternating maximization of individual variables.
Meanwhile, the objective function is the value of the LM rate between $X$ and $Z$, with its upper bound substantiated in \cite{Scarlett2020Information}. 
Consequently, the algorithm converges.
\end{proof}

\begin{remark}\label{rmk:01}
    An average power constraint like $\mathbb{E} \left[X^2\right] \leq \Gamma$ may be considered under the AWGN channel.    
    In this case, we can also design a similar Alternating Maximization algorithm based on the relaxation form \eqref{capacity_use}, following the analysis of its corresponding Lagrangian.
\end{remark}

\section{Numerical Results}  \label{sec_numerical}
This section applies the proposed AM algorithm to compute the mismatch capacity $C_d$ of the DM-IBC model.
All the experiments are conducted on a PC with 8G RAM and one Intel(R) Core(TM) Gold i5-8265U CPU @1.60GHz. 
\subsection{An Example over a Quaternary Channel}
First, we use the AM algorithm to compute \eqref{capacity_dis} over a Quaternary channel, and compare results with the MMIB algorithm of \cite{Dikshtein2023OnMO}.
The transition law is adopted the same settings with \cite[Eq.(8)]{Dikshtein2023OnMO}.
And the decoding rule $q(x,z) = e ^{-d(x,z)}$ is set to be 
\begin{equation*} 
q(x,z) =
\left\{  
 \begin{array}{ll}
 1 - \epsilon, & x = z, \\  
 \epsilon / 3, & x \neq z.   
 \end{array}  
\right.  
\end{equation*} 
Here the input distribution $P_X$ is assumed to be uniform,
and this yields the computational result $R_q$ is an achievable rate. 
The index $q$ represents the decoding rule $q(x,z)$.
\begin{figure}[H]
	\centering
		\includegraphics[width=0.90\linewidth]{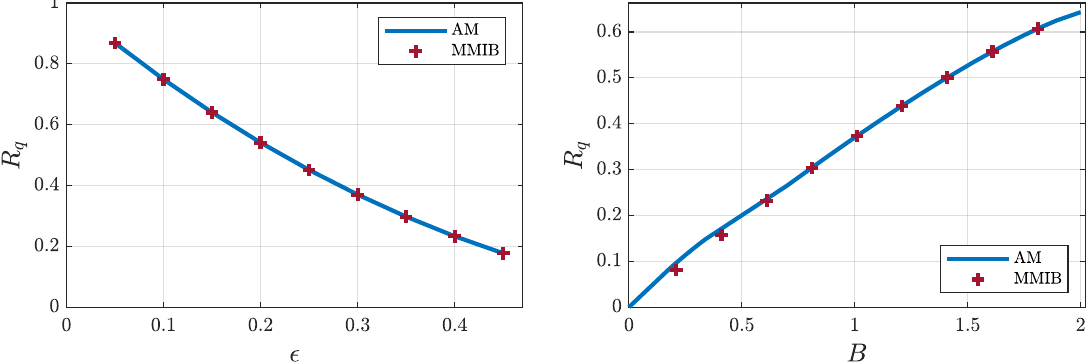}
        \caption{Left:
        $R_q$ versus channel parameter $\epsilon$ (as defined in \cite{Dikshtein2023OnMO}) with a fixed compression rate $B = 1.0$.
        Right:
        $R_q$ versus compression rate $B$ with a fixed channel parameter $\epsilon = 0.3$.}
	\label{fig:qual}
\end{figure}

The achievable rates $R_q$ for different cases of the Quaternary channel are demonstrated in Fig. \ref{fig:qual}.
In addition, Table \ref{tab:qual} compares the computational time of the AM algorithm and the MMIB algorithm. 
The algorithm terminates when two consecutive iterations yield a residual less than $10^{-5}$. 
It is obvious that the AM algorithm is superior to the MMIB algorithm in terms of the computational time. 
Further, the results
of the MMIB algorithm depend on the choice of search parameters and the initialization. 
However, these issues do not arise in the AM algorithm, which could converge stably with most random initial values.
%

\begin{table}[ht]
\centering
\begin{threeparttable}
\caption{Efficiency comparison between AM and MMIB algorithms}
\label{tab:qual}
\begin{tabular}{c|c|c|c|c|c}
    \hline
      \multirow{2}*{$(\epsilon, B)$} & \multicolumn{2}{c|}{Computational time (s)} & \multirow{2}*{\makecell{Speed-up\\ ratio}} & \multicolumn{2}{c}{Feasibility} \\
    \cline{2-3} \cline{5-6} 
      & $t_{\mathrm{MMIB}}$ & $t_{\mathrm{AM}}$ &  & MMIB & AM \\ 
    \hline 
    (0.3, 0.41) & 78 & 8.6 & 9 & 7e-9 & 8e-10 \\
    (0.3, 0.81) & 82 & 0.5 & 164 & 5e-9 & 1e-10 \\
    \hline
    (0.4, 0.41) & 118 & 1.6 & 74 & 7e-9 & 2e-10 \\
    (0.4, 0.81) & 134 & 7.5 & 18 & 6e-9 & 9e-10 \\
    \hline
\end{tabular}
\begin{tablenotes}
    \item 
    Columns 2-3 are the computational time of two algorithms. Columns 5-6 indicate whether the compression constraint is satisfied. 
\end{tablenotes}
\end{threeparttable}
\end{table}
\begin{remark}
In Table \ref{tab:qual}, $t_{\mathrm{AM}}$ denotes the total computational time, including the binary search for the corresponding multiplier $\lambda$,
since the AM algorithm requires a fixed multiplier, similar to that in the Blahut-Arimoto (BA) algorithm \cite{blahut_1972_computation}.
This does not affect the time advantage of the AM algorithm.
\end{remark}
\subsection{Examples over Gaussian Channels with IQ Imbalances}
\label{sec:num2}
With a further consideration of optimizing the input probability distribution $P_X$, applying a similar manner of the MMIB algorithm results in a more intricate search structure.
As the increasing number of discrete points in the alphabets, it may exceed the computational ability of a computer, due to the step of solving multiple roots of a nonlinear equations system.
Thus, we only adopt the AM algorithm in later experiments.
In this subsection, the channel between $X$ and $Y$ is configured as an additive white Gaussian noise (AWGN) channel with IQ imbalances,
modeled by
\begin{equation*}
Y = HX+N,\quad N\sim \mathcal{N}({\bf 0},\sigma_N^2 I).
\end{equation*}
The channel metric $H \in \mathbb{R}^{2 \times 2}$ is a combination of rotation and scaling effects as
\begin{equation*}
H = 
\begin{pmatrix}
    1 & 0\\
    0 & \eta
\end{pmatrix}
\begin{pmatrix}
    \cos\theta & \sin\theta \\
    -\sin\theta & \cos\theta
\end{pmatrix},
\end{equation*}
in which the parameter $\eta$, respectively, $\theta$ denotes the scaling and the degree of the rotation on the signal. 
Moreover, we set the signal-to-noise ratio $\mbox{SNR} \triangleq 10 \log _{10} \left({1}/{2\sigma_N^2}\right)$.
%

Under mismatched cases, the values of parameters $\eta$ and $\theta$ are unknown at the decoder.
The decoding metric is $d(\bm x,\bm z) \triangleq \|\bm z - \hat{H}\bm x\|_2^2$, which is the distance between the expected observation points and the actual reception points.
Here $\hat{H}$ is an estimate of $H$, and $\hat{H} \triangleq I$ in mismatch cases.
This work examines the QPSK and 16QAM modulation schemes.
The constellation points (i.e., the channel input alphabet $\mathcal{X}$) of all the examined modulation schemes are normalized as in \cite{ye2022optimal}, namely $\sum_{i=1}^M\frac{1}{M}\|x_i\|^2 = 1$. 
It is important to note that, for the input alphabet $\mathcal{X}$, the uniform distribution in the power constraint is solely employed to predetermine the coordinates of the collection points. 
When computing the mismatch capacity, the input distribution $P_X$ is unknown and is determined during the iterations. 
The constellation points for the post-compressed set $\mathcal{Z}$ are configured to be the same positions as those in $\mathcal{Y}$, the points before compression.
Note that alphabets $\mathcal{Y}$ and $\mathcal{Z}$ are continuous.
%
We restrict it to sufficiently large region, such as $[-8, 8]\times[-8, 8]$, and discretize it by
\begin{align*}
    z_{r\sqrt{N}+s} &= (-8+r\Delta z, -8+(s-1)\Delta z),\,\,\Delta z = \frac{16}{\sqrt{N}-1}, \\
    r &=0,1,\cdots ,\sqrt{N}-1,\,\, s = 1,2,\cdots, \sqrt{N},
\end{align*}
where $\{z_j\}_{j=1}^N$ is a set of uniform grid points corresponding to the alphabet $\mathcal{Z}$.
The number of discrete points in the alphabet of $\mathcal{Y}$ and $\mathcal{Z}$ are set to $N = 2500$, unless otherwise specified.
First, we examine the convergence property of the proposed algorithm by analyzing the residual errors
\footnote{As discussed in Remark \ref{rmk:01}, $\mu$ is the dual variable corresponding to the power constraint $\mathbb{E} \left[X^2\right] \leq \Gamma$ with the setting $\Gamma = 1$. And this constraint vanishes in the QPSK modulation scheme.}
\begin{subequations} \label{Err_norm}
\begin{align}
    r_{\phi} &= \sum_{i = 1}^M \left| \phi_i \sum\limits_{j=1}^N e^{-\zeta D_{ij}}\widetilde{\psi_j} P_Z(z_j) - P_X(x_i) \right|, \\
    r_{\widetilde{\psi}} &= \sum_{j = 1}^N 
    \left|
    \left(
    \widetilde{\psi_j}\sum\limits_{i=1}^M \phi_i e^{-\zeta D_{ij}} -1
    \right)P_Z(z_j)
    \right|, \\
    r_{\zeta} & = |G(\zeta)|,\quad r_{\mu}  = |F(\mu)|.
\end{align}
\end{subequations}

Fig. \ref{fig:modres} illustrates the convergent trajectories of residual errors versus iteration steps, with the parameters $(\eta, \theta)= (0.9,{\pi}/{18})$, SNR = 10 dB, and a fixed Lagrange multiplier $\lambda = 0.25$.
It suggests that the algorithm converges to $10^{-5}$ for all cases.
\begin{figure}[H]
	\centering
	\includegraphics[width=0.90\linewidth]{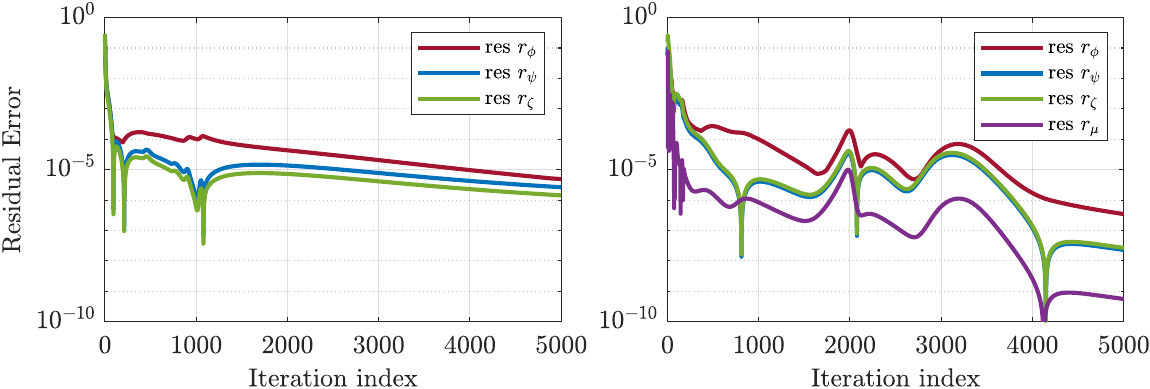}
        \caption{The convergent trajectories of the residual errors $r_{\phi}$ (Red), $r_{{\psi}}$ (Blue), $r_{\zeta}$ (Green), and $r_{\mu}$ (Purple). 
    Left: The QPSK modulation scheme.
    Right: The 16QAM modulation scheme.}
        \label{fig:modres}
\end{figure}
To facilitate an intuitive understanding of the relaying between $Y$ and $Z$, the values of the optimized conditional probability $P_{Z|Y}$ are plotted in Fig. \ref{fig:con}, with the discrete points $N = 225$ for clarity and simplify.
The vertical and horizontal coordinates correspond to the indices of the discrete points in alphabets $\mathcal{Y}$ and $\mathcal{Z}$, respectively. 
The results provide insight into the quantizer design at the relay node to some extent. 
Fig. \ref{fig:mis} compares the trade-off between the mismatch capacity $C_d$ and $I(Y;Z)$ in different settings, including a fixed parameter $(\eta, \theta) = (0.9, \pi/18)$ with different SNR values, 
and different mismatch cases
$(\eta, \theta)=$ $(0.9,{\pi}/{18})$, $(0.8,{\pi}/{18})$, $(0.9,{\pi}/{12})$, $(0.8,{\pi}/{12})$
with a fixed SNR = 10 dB.
The algorithm terminates after 5000 iterations, or when two consecutive iterations yield a residual less than $10^{-8}$.
As illustrated in Fig. \ref{fig:mis}, the channel capacity rises with increasing SNR values at the same compression rate.
The observed gap between the capacity of different mismatch cases becomes more noticeable with growing compression rates.

\begin{figure}[H]
    \centering
    \includegraphics[width=0.90\linewidth]{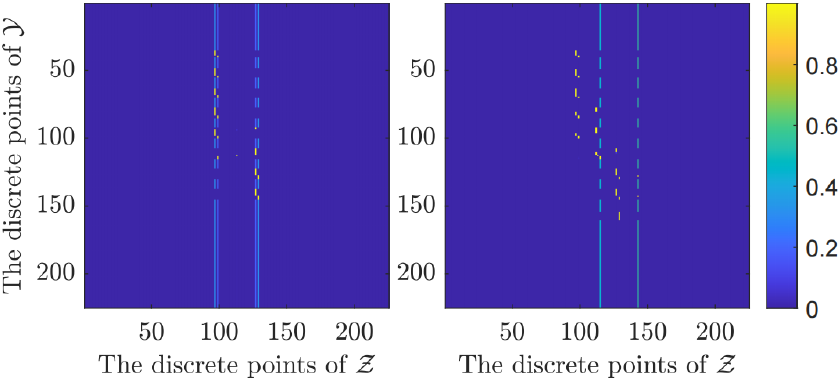}
    \caption{Values of the optimized conditional probability distribution $P_{Z|Y}$.
    Left: The QPSK modulation scheme.
    Right: The 16QAM modulation scheme.
    }
    \label{fig:con}
\end{figure}
\begin{figure}[H]
    \centering
    \includegraphics[width=0.90\linewidth]{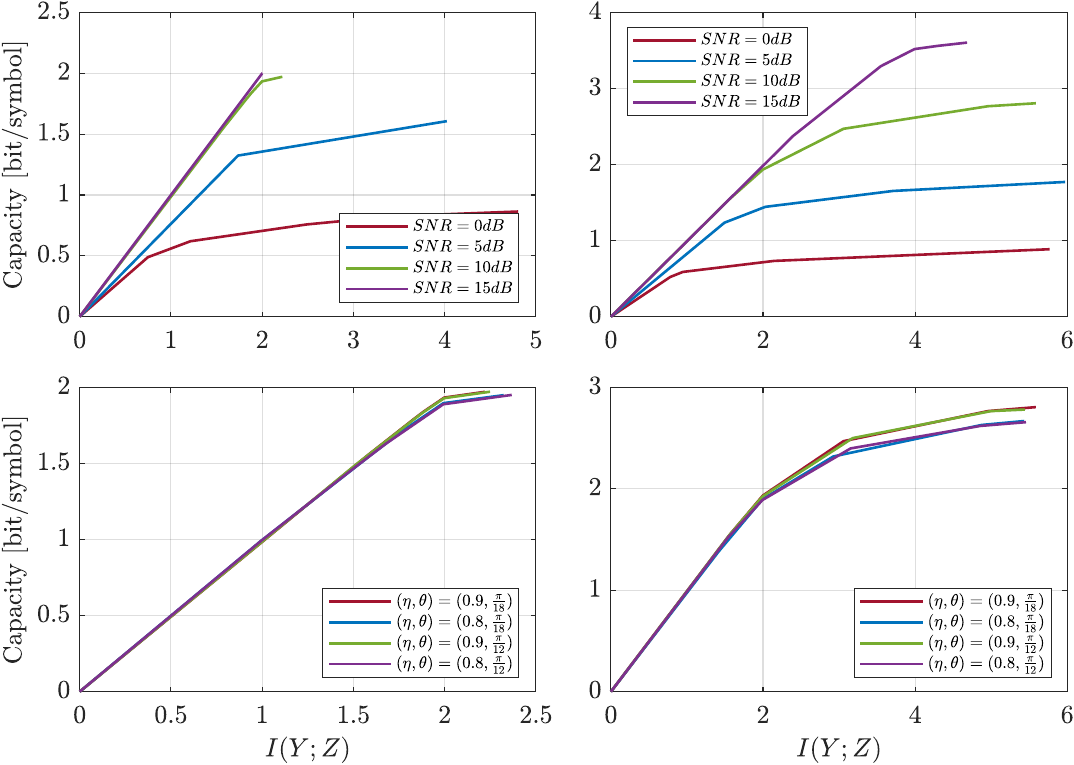}
    \caption{Mismatch capacity $C_d$ versus compression rate under different cases.
    Left: The QPSK modulation scheme.
    Right: The 16QAM modulation scheme.
    }
    \label{fig:mis}
\end{figure}%

\section{Conclusion}    \label{sec_conclusion}
This work presents an algorithm to evaluate the channel capacity of a 3-node point-to-point communication system with a relay.
Such a channel capacity takes the LM rate as the objective function, constrained by an IB-based problem, which makes the corresponding computation complicated.
To address the nonlinear constraint in the original optimization problem,
an IB-Lagrangian with a fixed multiplier is introduced.
We then proceed to examine the unconstrained relaxation form of the original optimization.
Further, a dual form of the LM rate from the perspective of the optimal transport is utilized, which enables the transformation of the problem into a consistent maximization form. 
By doing so, the alternating maximization is applicable, where each iteration step could be efficiently solved. 
Moreover, the convergence property of the proposed algorithm is proved. 
Finally, simulation results demonstrate the effectiveness of the proposed algorithm, and give promising application to practical scenarios like the quantizer design.
\newpage
\bibliographystyle{bibliography/IEEEtran}
\bibliography{bibliography/IBC_REF}

\newpage
\quad
\newpage 
\begin{appendix}

The detailed derivation of the AM algorithm will be presented in the following, as mentioned in Sec. \ref{sec:alg}.

\subsection{Update the input distribution $P_X$}
%

Consider the optimization problem \eqref{capacity_dis} with respect to $\bm p$ only, and recall its Lagrangian \eqref{lag_p}.
Taking the partial derivative of the Lagrangian $\mathcal{L}_p(\bm p; \eta_p)$ with respect to $\bm p$, we obtain
%
the optimal solution 
$$\widetilde{p_i} = J_i e^{-1 - \eta_p}.$$
In conjunction with the constraints $\sum_{i=1}^M p_i = 1$,
we can update $\bm p$ according to
\begin{equation}    \label{upd_p}
    p_i = \frac{J_i }{\sum_{i^{'} = 1}^M J_{i^{'}} },  
\end{equation}
where the dual variable $\eta_p$ is vanished.

\subsection{Update the joint distribution between $Y$ and $Z$}
%
%
First, consider the optimization problem with respect to $\Omega_{kj} = \widetilde{P}_{Z|Y}(z_j|y_k)$,
and construct its Lagrangian
\begin{equation}    \label{lag_t_a}
    \begin{aligned}
        \mathcal{L}_t (\bm \Omega; \bm \eta_{\bm t}) =
        & - \sum_{i=1}^M p_i \log p_i + \sum_{i = 1}^M p_i \log J_i(\bm\phi, \widetilde{\bm\psi}, \zeta; \bm \Omega, \bm r) \\
        & +1 
        - \sum_{k = 1}^K\eta_{t_k}\left( \sum_{j = 1}^N \Omega_{kj} - 1 \right), 
    \end{aligned}
\end{equation}
where $\{\eta_{t_k}\}_{k = 1}^K \subset \mathbb{R}$ are dual variables.
Taking the partial derivative of the Lagrangian \eqref{lag_t_a} with respect to $\Omega_{kj}$, we obtain
the optimal solution
\begin{equation*}
\begin{aligned}
   &  \widetilde{\Omega_{kj}} = 
    r_j *
    \exp \left(\frac{1}{\lambda}\left(
    -\sum_{i^{}=1}^M  \phi_{i^{}}e^{-\zeta D_{i^{}j}}\widetilde{\psi_j} 
    + \log \widetilde{\psi_j}
    \right) \right)\\
    &       
    * \exp \left(
    -\frac{\zeta\sum_{i=1}^M  D_{ij} \Theta_{ik}p_i }{\lambda\sum_{i=1}^M  \Theta_{ik} p_i}\right)
    * \exp \left( 1- 
    \frac{\eta_{t_k}}{\lambda\sum_{i=1}^M  \Theta_{ik} p_i}\right)
    .&
\end{aligned}
\end{equation*}
In conjunction with the constraints $\sum_{j=1}^N \Omega_{kj} = 1$, we could update $\Omega_{kj}$ according to
\begin{equation*}
    \Omega_{kj} = \frac{\Omega_{kj}^*}{\left(\sum_{j^{'}=1}^N \Omega_{kj^{'}}^*\right)},
\end{equation*}
where $\Omega_{kj}^*$ was defined in the main text.

Subsequently, consider the optimization problem with respect to $r_{j}$,
and construct its Lagrangian
\begin{equation}    \label{lag_r_a}
    \begin{aligned}
        \mathcal{L}_r (\bm r; \eta_r) = 
        & - \sum_{i=1}^M p_i \log p_i + \sum_{i = 1}^M p_i \log J_i(\bm\phi, \widetilde{\bm\psi}, \zeta; \bm \Omega, \bm r) \\
        & +1 
        - \eta_{r}\left(\sum_{j = 1}^N r_{j} - 1\right),
    \end{aligned}
\end{equation}
where $\eta_{r} \in \mathbb{R}$ is a dual variable.
Taking the partial derivative of the Lagrangian \eqref{lag_r_a} with respect to $r_{j}$, we obtain the optimal solution
\begin{equation}
    \widetilde{r_j} = \frac{\lambda}{\eta_{r}} \sum_{i=1}^M\sum_{k=1}^K\Omega_{kj} \Theta_{ik}p_i.
\end{equation}
In conjunction with the constraints in \eqref{capacity_dis},  we have
\begin{equation}
\begin{aligned}
    \sum_{j= 1}^N r_j = 1,\quad 
    \sum_{i=1}^M \sum_{k = 1}^K \sum_{j= 1}^N  \Omega_{kj} \Theta_{ik} p_i = 1.
\end{aligned}
\end{equation}
Thus $\eta_{r}  = \lambda$, which yields the updating of $r_j$
\begin{equation*}
    r_j =  \sum_{i=1}^M \sum_{k = 1}^K \Omega_{kj} \Theta_{ik} p_i.
\end{equation*}

\subsection{Update the dual form of the LM rate}

Similar to the previous work \cite{ye2022optimal}, we analyze the first-order condition of the variables $\bm \phi$, $\widetilde{\bm \psi}$, $\zeta$ respectively.
One could update the dual variables alternatively according to the following equations
\begin{equation*}
    \phi_i  = \dfrac{p_i}
    {\sum_{j=1}^N e^{-\zeta D_{ij}}\widetilde{\psi_j} r_j },\quad 
    \widetilde{\psi_j}  = \dfrac{1}
    {\sum_{i^{'}=1}^M \phi_{i^{'}} e^{-\zeta D_{i^{'}j}}}.
\end{equation*}

Also, the variable $\zeta$ could be updated by finding the root of a monotone function $G(\zeta)$ with Newton's Method, as defined in \eqref{upd_ze}.
%
Note that $G(0) \leq 0$ indicates that the corresponding constraint is satisfied. Thus, we can simply set $\zeta = 0$ without finding its root when $G(0) \leq 0$.
%


\subsection{AM algorithm with the power constraint}
An average power constraint like $\mathbb{E} \left[X^2\right] \leq \Gamma$ may be considered under the AWGN channel.
%
Accordingly, the optimization problem is turned to be \eqref{capacity_dis_a}, instead of \eqref{capacity_dis} in the main text.
\begin{equation}\label{capacity_dis_a}
\begin{aligned}
    \max_{\bm p, \bm t, \bm r}\,\,\max_{\substack{ \bdphi, \widetilde{\bdpsi}, \\ \zeta \geq 0}}& \,\,  - \bm p^T \log \bm p + \bm p^T \log \bm J(\bdphi, \widetilde{\bdpsi}, \zeta; \bm \Omega, \bm r) + 1 \\
    \mathrm{s.t.} &\,\,  \bm \Omega \mathbf{1}_{N} = \mathbf{1}_{K},\,\,
    \|\bm r\|_1 = 1
    ,\,  \|\bm p\|_1 = 1,
    \mathbb{E}_{\bm p} \|\bm x\| \leq \Gamma.
\end{aligned}
\end{equation}
To solve this problem, the alternating ascend technique is similarly adopted.
Actually, the difference only appears in the process of updating the input distribution $\bm p$.
The main derivation of the AM algorithm with the power constraint is given below.

Consider the optimization problem \eqref{capacity_dis_a} with respect to $\bm p$ only, and construct its Lagrangian
\begin{equation}    \label{lag_p_a}
    \begin{aligned}
        \mathcal{L}_p  (\bm p; \mu, \eta_p) = 
        - \sum_{i=1}^M p_i \log p_i + \sum_{i = 1}^M p_i \log J_i(\bm\phi, \widetilde{\bm\psi}, \zeta; \bm \Omega, \bm r) &\\
        +1 - \mu\left( \sum_{i = 1}^M p_i \|x_i\|^2 - \Gamma\right)
        - \eta_p \left( \sum_{i = 1}^M p_i - 1 \right),&
    \end{aligned}
\end{equation}
where $\mu \in \mathbb{R}^+$ and $\eta_p \in \mathbb{R}$ are dual variables.
%

Taking the partial derivative of the Lagrangian $\mathcal{L}_p(\bm p; \mu, \eta_p)$ with respect to $\bm p$, we obtain
%
the optimal solution 
$$\widetilde{p_i} = J_i e^{- \lambda \|x_i\|^2} e^{-1 - \eta_p}.$$
Then the dual form of \eqref{capacity_dis_a} with respect to $\bm p$ is constructed as
\begin{equation} \label{eq:dual21}
    \max_{ \lambda \geq 0, \,\, \eta_p} \,\,
    - e^{-1-\eta_p}
    \left(\sum_{i = 1}^M J_i e^{- \lambda \|x_i\|^2}\right)
    -1
    - \lambda\Gamma - \eta_p.
\end{equation}

By analyzing the dual optimization problem \eqref{eq:dual21}, we can update $\bm p$ according to
\begin{equation}    \label{upd_p_a}
    p_i = \frac{J_i e^{- \mu \|x_i\|^2}}{\sum_{i^{'} = 1}^M J_{i^{'}} e^{- \mu \|x_{i^{'}}\|^2}},  
\end{equation}
where the dual variable $\eta_p$ is vanished, and where $\mu \in \mathbb{R}^+$ is updated via finding the root of the following one-dimensional monotonic function
\begin{equation}    
    F(\mu) \triangleq -\Gamma + 
    \frac{\sum_{i = 1}^M \|x_i\|^2 J_i e^{- \mu \|x_i\|^2}}
    {\sum_{i = 1}^M J_i e^{- \mu \|x_i\|^2} }.
\end{equation}
When $F(0) \leq 0$, one could simply set $\mu = 0$ without finding its root, since the corresponding constraint has already been satisfied.

The AM algorithm with power constraint is summarized in the pseudo-code in Algorithm \ref{alg:main_2}.
\begin{algorithm}[htbp] 
	\renewcommand{\algorithmicrequire}{\textbf{Input:}}
	\renewcommand\algorithmicensure {\textbf{Output:} }
	\caption{AM Algorithm with power constraint}
	\label{alg:main_2}
	
	\begin{algorithmic}[1]
		\REQUIRE 
		Decoding metric $D_{ij}$, Conditional probability distribution $\Theta_{ik}$;  Power constraint coefficient $\Gamma$; The Lagrange multiplier $\lambda$; Iteration number $max\_iter$. \\
         
        \STATE \textbf{Initialize:} $\bm\phi^{(0)} = \bm 1_M$, $\widetilde{\bm\psi}^{(0)} = \bm 1_N$, $\zeta = 1$, $\mu = 1$; \textbf{Randomly Initialize:}  $\bm \Omega^{(0)}$, $\bm r^{(0)}$.\\

        \FOR{$l$ = 1 : $max\_iter$}
        
        \STATE  Solve $F(\mu) = 0$ for $\mu\in\mathbb{R}^+$ with Newton's method
        \STATE Update $p_i^{(l)}$ $\leftarrow$ ${J_i e^{- \mu \|x_i\|^2} }/{\left(\sum_{i^{'} = 1}^M J_{i^{'}} e^{- \mu \|x_{i^{'}}\|^2}\right)}$
        
        \STATE Update $\Omega_{kj}^{(l)}$ $\leftarrow$ ${\Omega_{kj}^*} / {\left(\sum_{j^{'}=1}^N \Omega_{kj^{'}}^*\right)}$

        \STATE Update $r_j^{(l)}$ $\leftarrow$ $\sum_{i=1}^M \sum_{k = 1}^K  \Omega_{kj}^{(l)} \Theta_{ik} p_i^{(l)}$

        \STATE Update $\phi_i^{(l)}$ $\leftarrow$ 
        ${p_i}^{(l)} / 
        \left({\sum_{j=1}^N e^{-\zeta D_{ij}}\widetilde{\psi_j}^{(l-1)} r_j^{(l)}}\right)$
      
        \STATE Update $\widetilde{\psi_j}^{(l)}$ $\leftarrow$ 
        ${1}/
        \left({\sum_{i=1}^M \phi_i^{(l)} e^{-\zeta D_{ij}}}\right)$
      
        \STATE  Solve $G(\zeta) = 0$ for $\zeta\in\mathbb{R}^+$ with Newton's method
        
        \ENDFOR
        \RETURN $C_d$
	\end{algorithmic}
\end{algorithm}

In the numerical simulations of Sec. \ref{sec:num2}, a power constraint $\mathbb{E} \left[X^2\right] \leq \Gamma$ is taken into consideration.
And we set $\Gamma = 1$ according to the 
predetermined coordinate positions of the collection points.
Particularly, in the QPSK modulation scheme, 
each collection point shares the same distance $\|\bm x\| = 1$ according to the initialization.
Thus, the power constraint here is equivalent to $\sum_{i=1}^N p_i \leq 1$, which always holds by the property of the probability distribution.

\end{appendix}

\end{document}